%% file: jumbled.tex
\documentclass[11pt]{llncs}

\usepackage{fullpage}
\usepackage{epsfig}
\usepackage{graphics}
\usepackage{amsmath,setspace}

\usepackage[dvipsnames,usenames]{color}
\usepackage[colorlinks=true,urlcolor=Blue,citecolor=Green,linkcolor=BrickRed]{hyperref}
\newcommand{\Oh}[1]
    {\ensuremath{{O}({#1})}}

\begin{document}

\title{Binary Jumbled Pattern Matching\\ on Trees and Tree-Like Structures\thanks{Preliminary version of this paper appeared in the 21st Annual European Symposium on Algorithms (ESA 2013)}}

\author{Travis Gagie\inst{1}   \and Danny Hermelin \inst{2}  \and Gad M. Landau\inst{3}\thanks{Supported in part by the National Science Foundation (NSF) grant 0904246, the Israel Science Foundation (ISF) grant 347/09,
and the United States-Israel Binational Science Foundation (BSF) grant 2008217 }  \and Oren Weimann\inst{3}\thanks{Supported in part by the Israel Science Foundation grant 794/13}  }
\institute{
University of Helsinki, \href{mailto:travis.gagie@aalto.fi}{travis.gagie@aalto.fi} \and
Ben-Gurion University, \href{mailto:hermelin@bgu.ac.il}{hermelin@bgu.ac.il} \and
University of Haifa, \href{mailto:landau@cs.haifa.ac.il}{landau@cs.haifa.ac.il}, \href{mailto:oren@cs.haifa.ac.il}{oren@cs.haifa.ac.il}  }

\maketitle

\begin{abstract}
Binary jumbled pattern matching asks to preprocess a binary string $S$ in order to answer queries $(i,j)$ which ask for a substring of $S$ that is of length $i$ and  has exactly $j$ 1-bits. This problem naturally generalizes to vertex-labeled trees and graphs by replacing ``substring'' with ``connected subgraph''.
In this paper, we give an $O(n^2 / \log^2 n)$-time solution for trees, matching the currently best bound for (the simpler problem of) strings. We also give an $\Oh{g^{2 / 3} n^{4 / 3}/(\log n)^{4/3}}$-time solution for strings that are compressed by a grammar of size $g$. This solution improves the known bounds when the string is compressible under many popular compression schemes. Finally, we prove that on graphs the problem is fixed-parameter tractable with respect to the treewidth $w$ of the graph, even for a constant number of different vertex-labels, thus improving the previous best $n^{O(w)}$ algorithm.\end{abstract}

\input{intro}
\input{trees.tex}

\input{grammars.tex}
\input{treewidth}

\section{Conclusions and Open Problems}
In this paper we considered the binary jumbled pattern matching problem on trees, bounded treewidth graphs, and strings compressed by grammars. We gave an $\tilde O(g^{2 / 3} n^{4 / 3})$-time solution for strings of length $n$ represented by grammars of size $g$, an $f(w)\cdot n^{O(1)}$-time solution for graphs with treewidth $w$, and an $O(n^2 / \log^2 n)$-time solution for trees. %Our time bounds for trees match the known bounds for the problem on strings.
In the latter result, we showed how to determine in $O(1)$ time if a query pattern appears, and how to locate in $O(\log n)$ time a node of this appearance. With a linear-space solution, locating the entire appearance remains an open problem.
Using Lemma~\ref{lemma:pms}, the construction time for trees can be made $O(n\cdot i/\log^2 n)$ if the query patterns are known to be of size at most $i$. We also note here that the construction time can be made faster on trees that have many identical rooted subtrees. This is because the bottom-up construction does not need to be applied on the same subtree twice.

Finally,  the main open problems stemming from our work is: (1) To obtain a faster construction of the linear-space index for strings. Our index for trees implies that any construction speedup for  strings implies a construction speedup for trees. (2) To develop an algorithm for the non-indexing variant of binary jumbled pattern matching on trees whose performance is closer to the performance of the corresponding algorithm on strings (\emph{i.e.} the $O(n)$ sliding window algorithm).

\section{Acknowledgments}
We thank the anonymous reviewers for their  helpful comments.

\bibliographystyle{plain}
\bibliography{jumbled}
\end{document}

%% file: intro.tex
% !TEX root = jumbled.tex
\section{Introduction}
\label{sec:introduction}

\emph{Jumbled pattern matching} is an important variant of classical pattern matching with several applications in computational biology, ranging from alignment~\cite{Alignment} and  SNP discovery~\cite{SNP}, to the interpretation of mass spectrometry data~\cite{BCFL10} and metabolic network analysis~\cite{LFS06}. In the most basic case of strings, the problem asks to determine whether a given pattern $P$ can be rearranged so that it appears in a given text $T$. That is, whether $T$ contains a substring of length $|P|$ where each letter of the alphabet occurs the same number of times as in $P$.
Using a straightforward sliding window algorithm, such a jumbled occurrence can be found optimally in $O(n)$ time on a text of length $n$. While jumbled pattern matching has a simple efficient solution, its {\em indexing} problem is much more challenging. In the indexing problem,
we preprocess a given text $T$ so that on queries $P$ we can determine quickly whether $T$ has a jumbled occurrence of $P$. Very little is known about this problem besides the trivial naive solution.

Most of the interesting results on indexing for jumbled pattern matching relate to binary strings (where a query pattern $(i,j)$ asks for a substring of $T$ that is of length $i$ and has $j$ 1s).
Given a binary string of length $n$, Cicalese, Fici and Lipt\'ak~\cite{CFL09} showed how one can build in $O(n^2)$ time an $O(n)$-space index that answers jumbled pattern matching queries in $O(1)$ time.
Their key observation was that if one substring of length $i$ contains fewer than $j$ 1s, and another substring of length $i$ contains more than $j$ 1s, then there must be a substring of length $i$ with exactly $j$ 1s. Using this observation, they construct an index that stores the maximum and minimum number of 1s in any $i$-length substring, for each possible~$i$.
Burcsi \emph{et al.}~\cite{BCFL10} (see also~\cite{BCFL12a,BCFL12b}) and Moosa and Rahman~\cite{MR10} independently improved the construction time to $O(n^2 / \log n)$, then Moosa and Rahman~\cite{MR12} further improved it to $O(n^2 / \log^2 n)$ in the word RAM model.
Currently, faster algorithms than $O(n^2 / \log^2 n)$ exist only when the string compresses well using run-length encoding~\cite{BFKL12,GG12} or when we are willing to settle for approximate indices~\cite{CLWY12}.
Regarding non-binary alphabets, the recent solution of Kociumaka, Radoszewski and Rytter~\cite{ESA13}  for constant alphabets requires  $o(n^2)$ space and $o(n)$ query
time. For alphabets of size
$\omega(1)$, sublinear query time was achieved by Burcsi et al.~\cite{BCFL12a} for large query patterns but in the worst case a query takes superlinear time. In fact, a recent result of Amir et al.~\cite{Amir} shows that under the popular 3-SUM conjecture, jumbled indexing for alphabets of size $\omega(1)$ requires either $\Omega(n^{2-\varepsilon})$ preprocessing time or $\Omega(n^{2-\delta})$ query time for any $\varepsilon, \delta >0$.

The natural extension of jumbled pattern matching from strings to trees is much harder. In this extension, we are asked to determine whether a vertex-labeled input tree has a connected subgraph where each label occurs the same number of times as specified by the input query. The difficulty here stems from the fact that a tree can have an exponential number of connected subgraphs as opposed to strings. Hence, a sliding window approach becomes intractable. Indeed, the problem is NP-hard~\cite{LFS06}, even if our query contains at most one occurrence of each letter~\cite{FFHV11}. It is
not even fixed-parameter tractable when parameterized by the alphabet size~\cite{FFHV11}.
The fixed-parameter tractability of the problem was further studied when extending the problem from trees to graphs~\cite{AmbalathBHKMPR10,BetzlerBFKN11,DondiFV11,DondiFV11b}. In particular, the problem (also known as the {\em graph motif problem}) was recently shown by Fellows \emph{et al.}~\cite{FFHV11} to be polynomial-time solvable when the number of letters in the alphabet as well as the treewidth of the  graph are both fixed. They also gave an fixed-parameter algorithm when the size of the pattern is taken as a parameter, and showed that no such algorithm is likely to exist when the problem is parameterized by the alphabet size, even in case the input graph is a tree. This latter result implies that assuming the Exponential Time Hypothesis (ETH), there is no $n^{o(\sqrt{|\Sigma|})}$ algorithm for jumbled pattern matching on trees over general alphabets $\Sigma$.

 \paragraph{\bf Our results.}
In this paper we extend the currently known state-of-the-art for binary jumbled pattern matching. Our results focus on trees, and tree-like structures such as grammars and bounded treewidth graphs. We use the word RAM model of computation with the standard assumption that the word-length is at least $\log n$.

\begin{itemize}
\item[$\bullet$] {\bf Trees:} For a tree $T$ of size $n$, we present an index of size $O(n)$ bits that is constructed in
$O(n^2 / \log^2 n)$ time and answers binary jumbled pattern matching queries in $O(1)$ time. This matches the performance of the best known index for binary strings. In fact, our index for  trees is obtained by multiple applications of an efficient algorithm for strings~\cite{MR12} under a more careful analysis. This is combined with both a micro-macro~\cite{MicroMacro} and centroid decomposition of the input tree.
Our index can also be used as an $O(ni/\log^2 n)$-time algorithm for the pattern matching (as opposed to the indexing) problem, where $i$ denotes the size of the pattern. Finally, by increasing the space of our index to $O(n \log n)$ bits, we can output in $O(\log n)$ time a node of $T$ that is part of the pattern occurrence.\\

\item[$\bullet$] {\bf Grammars:} For a binary string $S$ of length $n$ derived by a grammar of size $g$, we show how to construct in $O(g^{2/3} n^{4/3}/\log^{4/3} n)$ time an index of size $O(n)$ bits that answers jumbled pattern matching queries on $S$ in $O(1)$ time.
The size of the grammar $g$ can be exponentially smaller than $n$ and is always at most $O(n/\log n)$. This means that our  time bound is $O(n^2 / \log^2 n)$ even when $S$ is not compressible. If $S$ is compressible but with other compression schemes such as the LZ-family, then we can transform it into a grammar-based
compression with little or no expansion~\cite{Charikaretal2005,Ryt03}.\\

\item[$\bullet$] {\bf Bounded Treewidth Graphs:}  For a graph $G$ with treewidth bounded by $w$, we show how to improve on the $O(n^{O(w)})$ time algorithm of Fellows \emph{et al.}~\cite{FFHV11} to an algorithm which runs in $2^{O(w^3)}n+w^{O(w)}n^{O(1)}$ time. Thus, we show that for a binary alphabet, jumbled pattern matching is fixed-parameter tractable when parameterized only by the treewidth. This result extends easily to alphabets of constant sizes.
\end{itemize}

%\noindent Our results are organized in the paper as follows: section 2 deals with trees, section 3 deals with grammars, and finally, section 4 deals with bounded treewidth graphs.
\noindent We present our results for trees, grammars, and bounded treewidth graphs in sections 2, 3 and 4 respectively.

%% file: trees.tex
% !TEX root = jumbled.tex
\section{Jumbled Pattern Matching on Trees}
\label{sec:motifs}

In this section we consider the natural extension of binary jumbled pattern matching to trees. Recall that in this extension we are given a tree $T$ with $n$ nodes, where each node is labeled by either 1 or 0. We will refer to the nodes labeled 1 as \emph{black} nodes, and the nodes labeled 0 as \emph{white} nodes. Our goal is to construct a data structure that on query $(i,j)$ determines whether $T$ contains a connected subgraph with exactly $i$ nodes, $j$ of which are black. Such a subgraph of $T$ is referred to as a \emph{pattern} and $(i,j)$ is said to {\em appear} in $T$. The main result of this section is stated below.

\begin{theorem}
\label{thm:motifs}
Given a tree $T$ with $n$ nodes that are colored black or white, we can construct in $O(n^2/\log^2 n)$ time a data structure of size $O(n)$ bits that given a query $(i,j)$ determines in $O(1)$ time if  $(i,j)$ appears in $T$.
\end{theorem}

\noindent Notice that the bounds of Theorem~\ref{thm:motifs} match the currently best bounds for the case where $T$ is a string~\cite{MR10,MR12}. This is despite the fact that a string has only $O(n^2)$ substrings while a tree can have $\Omega(2^n)$ connected subgraphs. The following lemma indicates an important property of string jumbled pattern matching that carries on to trees. It gives rise to a simple index described below.

\begin{lemma}
\label{lemma:minmax}
If  $(i,j_1)$ and $(i,j_2)$ both appear in $T$, then for every $j_1\le j \le j_2$,  $(i,j)$ appears in $T$.
\end{lemma}

\begin{proof}
Let $j$ be an arbitrary integer with $j_1 \leq j \leq j_2$, and let $T_1$ and $T_2$ be two patterns in $T$ corresponding to $(i,j_1)$ and $(i,j_2)$ respectively. The lemma follows from the fact that there exists a sequence of patterns starting with $T_1$ and ending with $T_2$ such that every pattern has exactly $i$ nodes and two consecutive patterns differ by removing a leaf from the first pattern and adding a different node instead. This means that the number of black nodes in two consecutive patterns differs by at most 1.
\qed
%Observe that constructing such a sequence is not difficult when $T_1$ and $T_2$ share at least one node. When the two trees are disjoint, the sequence can be constructed by identifying a path between a node of $T_1$ and a node of $T_2$ and then constructing the sequence by replacing leaves with nodes from the path until we reach a tree which shares a common node with $T_2$.
\end{proof}

\subsection{\bf A Simple Index}
As in the case of strings, the above lemma suggests an $O(n)$-size data structure: For every $i=1,\ldots,n$, store the minimum and maximum values $i_{min}$ and $i_{max}$ such that  $(i,i_{min})$ and $(i,i_{max})$ appear in $T$. This way, upon query $(i,j)$, we can report in constant time whether $(i,j)$ appears in $T$ by checking if $i_{min} \le j \le i_{max}$. However, while $O(n^2)$ construction-time is trivial for strings (for every $i=0,\ldots,n$, slide a window of length $i$ through the text in $O(n)$ time) it is harder on trees.

To obtain $O(n^2)$ construction time, we begin by converting our tree into a rooted binary tree. We arbitrarily root the tree $T$. To convert it to a binary tree, we duplicate each node with more than two children as follows: Let $v$ be a node with children $u_1, \ldots,u_k$, $k \geq 3$. We replace $v$ with $k-1$ new nodes $v_1,\ldots,v_{k-1}$, make $u_1$ and $u_2$ be the children of $v_1$, and make $v_{\ell-1}$ and $u_{\ell+1}$ be the children of $v_{\ell}$ for each $\ell=2,\ldots,k-1$. If $v$ is not the root then we set the parent of $v_{k-1}$ to be the parent of $v$ (otherwise, $v_{k-1}$ is the root). The node $v_1$ gets the same color as the the node $v$. The other nodes $v_2,\ldots,v_k$ are called \emph{dummy nodes} and have no color. This procedure at most doubles the size of $T$. To avoid cumbersome notation, we henceforth use $T$ and $n$ to denote the resulting  rooted binary tree and its number of nodes respectively. For a node $v$, we let $T_v$ denote the subtree of $T$ rooted at $v$ (\emph{i.e.} the connected subgraph induced by $v$ and all its descendants).

Next, in a bottom-up fashion, we compute for each node $v$ of $T$ an array $A_v$ of size $|T_v|+1$. The entry $A_v[i]$ will store the maximum number of black nodes that appear in a connected subgraph of size $i$ that includes $v$ and another $i-1$ nodes in $T_v$. Computing the minimum (rather than maximum) number of black nodes is done similarly. %The array $A_v$ for leafs $v$ is set to zero except for the entry $A_v[1]$ which is set to 1 if $v$ is black.
Throughout the execution, we also maintain a global array $A$ such that $A[i]$ stores the maximum $A_v[i]$ over all nodes $v$ considered so far. Notice that in the end of the execution, $A[i]$ holds the desired value $i_{max}$ since every connected subgraph of $T$ of size $i$ includes some node $v$ and $i-1$ nodes in $T_v$.

We now show how to compute $A_v[i]$ for a node $v$ and a specific value $i \in \{1,\ldots,|T_v|\}$. If $v$ has a single child $u$, then $v$ is necessarily not a dummy node and we set $A_v[i]=col(v)+A_u[i-1]$, where $col(v)=1$ if $v$ is black and $col(v)=0$ otherwise. If $v$ has two children $u$ and $w$, then any pattern of size $i$ that appears in $T_v$ and includes $v$ is composed of $v$, a pattern of size $\ell$ in $T_u$ that includes $u$, and a pattern of size $i-1-\ell$ in $T_w$ that includes $w$. We therefore set  $A_v[i]= col(v)+\max_{0 \leq \ell \leq i-1} \{A_u[\ell]+A_w[i-1-\ell]\}$ and
$A_v[i]=\max_{1 \leq \ell \leq i-1}\{A_u[\ell]+A_w[i-1-\ell]\}$ when $v$ is a dummy node. Observe that in the latter the index $\ell$  starts with 1 to indicate that the non-dummy copy of $v$ (i.e., $v_1$) must be included in the pattern. 

We next analyze and then improve the above algorithm (first by one log factor and then by another log factor). In the rest of this section, like the above algorithm, all of our algorithms will compute the $A_v$ arrays for each $v$ in $T$ in a bottom-up fashion. In all these algorithms, just like in the above algorithm, a special attention has to be given to the case where $v$ is a dummy node. Handling dummy nodes is done similarly to the above. To make the presentation simpler we will assume that there are no dummy nodes at all.  

\begin{lemma}
\label{lem:quardratic}%
The above algorithm runs in $O(n^2)$ time.
\end{lemma}

\begin{proof}
The computation done on nodes with one child requires $O(n)$ time, hence the total time required to compute all arrays $A_v$ for such nodes is $O(n^2)$. The time required to compute all arrays for nodes with two children is asymptotically bounded by the sum $\sum_v \alpha(v)\beta(v)$, where $\alpha(v)$ and $\beta(v)$ denote the sizes of the two subtrees rooted at each of the children of $v$, and the sum is taken over all nodes $v$ with two children. For a tree rooted at $r$, we let $cost(r)$ denote this sum over all nodes in $T_r$ and argue by induction that $cost(r)$ is bounded by $|T_r|^2 = O(n^2)$.

Let $r$ be the root of a tree with $n$ nodes, and let $u$ and $v$ denote the two children of $r$. Let $x$ denote the size of the subtree rooted at $u$. Then $x < n$, and the size of the subtree rooted at $v$ is $n-1-x$. By induction, we have $cost(u) \leq x^2$ and $cost(v) \leq (n-1-x)^2$. Thus, $
cost(r)  =  x(n-1-x) + cost(u) + cost(v)   <  n^2-x(n-x)   \leq  n^2 $.
 \qed
\end{proof}

Note that if at any time the algorithm only stores arrays $A_v$ which are necessary for future computations, then the total space used by the algorithm is $O(n)$. The space can be made $O(n)$ {\em bits} by storing the $A_v$ arrays in a succinct fashion (this  will also prove useful later for improving the running time): Observe that $A_v[i+1]$ is either equal to $A_v[i]$ or to $A_v[i]+1$. This is because any pattern of size $i$ with $b$ black nodes can be turned into a pattern of size $i-1$ with at least $b-1$ black nodes by removing a leaf. We can therefore represent $A_v$ as a binary string $B_v$ of $n+1$ bits, where $B_v[0]=0$, and $B_v[i]= A_v[i]-A_v[i-1]$ for all $i=1,\ldots,n$. Notice that  since $A_v[i]= \sum_{\ell=0}^i B_v[\ell]$, each entry of $A_v$ can be retrieved from $B_v$ in $O(1)$ time using {\em rank} queries~\cite{Jacobson,Munro}.

\input{pattern_matching}
\input{4russians}

\input{listing}

%% file: pattern_matching.tex
% !TEX root = jumbled.tex
\subsection{Pattern Matching}
\label{matching}

Before improving the above algorithm, we show that it can already be analyzed more carefully to get a bound of $O(n\cdot i)$ when the pattern size is known to be at most $i$. This means that in $O(n)$ space and $O(n\cdot i)$ construction time  we can build an index that answer queries in $O(1)$ time provided the pattern size is bounded by $i$. 
It is also useful for the {\em pattern matching} problem: without preprocessing, decide whether a given pattern $(i,j)$ appears in $T$.

In the case of strings, this problem can trivially be solved in $O(n)$ time by sliding a window of length $i$ through the string thus effectively considering every substring of length $i$. This sliding-window approach however does not extend to trees since we cannot afford to examine all connected subgraphs of $T$. We next show that, in trees,  searching for a pattern of size $i$ can be done in  $O(n\cdot i)$ time by using our above indexing algorithm. This is useful when the pattern is small (i.e., when $i=o(n)$). Obtaining $O(n)$ time remains our main open problem.

\begin{lemma}
\label{lemma:pms}
Given a tree $T$ with $n$ nodes that are colored black or white and a query pattern $(i,j)$,
we can check in $O(n\cdot i)$ time and $O(n)$ space if $T$ contains the pattern $(i,j)$.
\end{lemma}
\begin{proof}
In our indexing algorithm, every node $v$ computes an array $A_v$ of size $|T_v|$.
When the pattern size is known to be $i$ we can settle for an array $A_v$ of size $\min\{|T_v|,i\}$. Recall from the above discussion that we can assume $T$ is a binary tree. Consider some node $v$ that has only one child $u$. We can  compute $A_v$ from $A_u$ in time $O(\min\{|T_v|,i\})=O(i)$. Summing over all such nodes $v$ gives at most $O(n\cdot i)$.
If on the other hand,  node $v$ has two children $u$ and $w$ then $A_v$ is computed from $A_u$ and $A_w$ in $O(\min\{|T_{u}|, i\} \cdot \min\{|T_{w}|, i\})$ time. We claim that summing this term over all nodes in $T$ that have two children gives
$O(n\cdot i)$.

To see this, first consider the subset of nodes $V = \{v\in T :  |T_v|< i \text{ and } |T_{parent(v)}|\ge i\}$, where $parent(v)$ denotes the parent of $v$ in $T$. Notice that each subtree $T_v\in \{T_{v} : v\in V\}$ is of size less than $i$ and that these subtrees are disjoint. By the proof of Lemma~\ref{lem:quardratic} we know that computing $A_v$ (along with every $A_u$ for vertices $u\in T_v$) is done in $O(|T_v|^2)$ time. The total time to compute $A_v$ for all nodes $v\in V$ and their descendants is therefore $cost(v) = \sum_{v\in V} |T_v|^2$. Since every $|T_v|< i$ and $\sum_{v\in V} |T_v| \le n$, we have that $cost(v)$ is upper bounded by $O(n\cdot i)$ that is achieved when all $|T_v|s$ are equal to $i$ and $|V|=n/i$.

The remaining set of nodes $S$ consists of all nodes $v$ such that $v$ has two children $u,w$ and $|T_v|\ge i$. We partition these nodes into $S_1=\{v\in S : |T_u| \ge i \text{ and } |T_w|  \ge i \}$ and $S_2 = S \setminus S_1$. Notice that $|S_1| = O(n/i)$. % since any tree can have at most $n/i$ nodes that have two children with subtrees larger than $i$.
Therefore, computing $A_v$ for all nodes $v\in S_1$ can be done in $O(|S_1|\cdot i^2)=O(n\cdot i)$ time.
We are left only with the vertices of $S_2$. These are all vertices $v$ such that at least one of their children is in $V$. Denote this child as $d(v)$. Computing $A_v$ for all nodes in $S_2$ can therefore be done in time
$$\sum_{v\in S_2}O(|T_{d(v)}|\cdot i)  =    i\cdot \sum_{v\in S_2}O(|T_{d(v)}|) =  i \cdot \sum_{u\in V}O(|T_{u}|) =  O(i \cdot n).$$ \qed
\end{proof}

%% file: 4russians.tex
% !TEX root = jumbled.tex
\subsection{An Improved Index}
\label{4russians}

In this subsection, we will gradually improve the construction time from $O(n^2)$ to $O(n^2/\log^2 n)$. For simplicity of the presentation, we will assume the  input tree $T$ is a rooted binary tree. This extends to arbitrary trees using a similar dummy-nodes trick as above.

\paragraph{\bf From trees to strings.} Recall that we can represent every $A_v$ by a binary string $B_v$ of $n+1$ bits where $B_v[0]$ is always zero and for $i=1,\ldots,n$,  $B_v[i]= A_v[i]-A_v[i-1]$.  We begin by showing that if $v$ has two children $u,w$ then the computation of $B_v$ can be done by solving a variant of jumbled pattern matching on the string $S_v= X_v\circ col(v) \circ Y_v$ (here $\circ$ denotes concatenation) of length $|S_v| = |T_u| + |T_w| +1$, where $X_v$ is obtained from $B_u$ by reversing it and removing its last bit, and $Y_v$ is obtained from $B_w$ by removing its first bit. We call the position in $S_v$ with $col(v)$ the \emph{split position} of $S_v$. Recall that $A_v[i]= col(v)+\max_{0 \leq \ell \leq i-1} \{A_u[\ell]+A_w[i-1-\ell]\}$. This is equal to the maximum number of 1s in a window of $S_v$ that is of length $i$ and includes the split position of $S_v$.

We are therefore interested only in windows including the split position, and this is the important distinction from the standard jumbled pattern matching problem on strings. Clearly, using the fastest $O(n^2/\log^2 n)$-time algorithm~\cite{MR12} for the standard string problem we can also solve our problem and compute $A_v$ in $O(|S|^2/\log^2 n)$ time. However, recall that for our total analysis (over all nodes $v$) to give $O(n^2/\log^2 n)$ we need the time to be $O(|X_v|\cdot|Y_v| /\log^2 n)$ and not $O((|X_v| + |Y_v|)^2 /\log^2 n)$.

\paragraph{\bf First speedup.} The $O(\log^2 n)$-factor speedup for jumbled pattern matching on strings~\cite{MR12} is achieved by a clever combination of lookup tables (also known as the ``Four Russians technique'') . One log factor is achieved by computing the maximum number of 1s in a window of length $i$ only when $i$ is a multiple of $s = (\log n)/6 $. Using a lookup table over all possible pairs of length-$s$ windows, a sliding window of size $i$ can be extended in $O(1)$ time to all windows of sizes $i+1,\ldots,i+ s -1$ that start at the same location (see~\cite{MR12} for details). Their algorithm can output in $O(n^2/ \log n)$ time an array of $O(n/\log n)$ words. For each $i$ that is a multiple of $s$, the array keeps one word storing the maximum number of 1s over all windows of length $i$ and another word storing the binary increment vector for the maximum number of 1s in all windows of length $i+1,\ldots,i+ s -1$.

By only considering windows that include the split position of $S_v$, this idea easily translates to an $O(|X_v|\cdot|Y_v| /\log n)$-time algorithm to compute $A_v$ and implicitly store it in $O((|X_v|+|Y_v|) /\log n)$ words. From this it is also easy to obtain an $O((|X_v|+|Y_v|) /\log n)$-words representation of $B_v$. Notice that if $v$ has a single child then the same procedure works with $|X_v|=0$ in time $O(|Y_v| /\log n)=O(n/\log n)$. Summing over all nodes $v$, we  get an $O(n^2/\log n)$-time solution for binary jumbled indexing on trees.

\paragraph{\bf Second speedup.}

In strings, an additional logarithmic improvement shown in~\cite{MR12} can be obtained as follows: When sliding a window of length $i$ ($i$ is a multiple of $s$) the window is shifted $s$ locations in $O(1)$ time using a lookup table over all pairs of binary substrings of length $s$ (representing the leftmost and rightmost bits in all these $s$ shifts). This further improvement yields an $O(n^2 /\log ^2 n)$-time algorithm for strings. In trees however this is not the case. While we can compute $A_v$ in $O((|X_v|+|Y_v|)^2 /\log ^2 n)$ time, we can guarantee $O(|X_v|\cdot|Y_v| /\log^2 n)$ time only if both $|X_v|$ and $|Y_v|$ are greater than $s$. Otherwise, say $|X_v|<s$ and $|Y_v|\ge s$, we will get $O(|X_v|\cdot|Y_v| / |X_v| \log n ) = O(|Y_v| /\log n)$ time.
This is because our windows must include the $col(v)$ index and so we never shift a window by more than $|X_v|$ locations.
Overcoming this obstacle is the main challenge of this subsection. It is achieved by carefully ensuring that the $O(|Y_v| /\log n) = O(n /\log n)$ costly constructions will be done only  $O(n /\log n)$ times.

%We begin by extending the notion of $A_v$ so that a pattern will not only include vertices of $T_v$ but in some cases as much as $x < \log n$ vertices that are not in $T_v$. Consider the linear-time {\em micro-macro} decomposition~\cite{MicroMacro} of $T$.

\paragraph{\bf A micro-macro decomposition.} A {\em micro-macro decomposition}~\cite{MicroMacro} is a partition of $T$ into
$O(n/\log n)$ disjoint connected subgraphs called {\em micro trees}. Each micro tree is of size at most $\log n$, and at most two nodes in a micro tree are adjacent to nodes in other micro trees.
These nodes are referred to as \emph{top} and \emph{bottom boundary} nodes. The top boundary node is chosen as the root of the micro tree.
The {\em macro tree} is a rooted tree of size $O(n/\log n)$ whose nodes correspond to micro trees as follows (See Fig.~\ref{fig:micromacro}):  The top boundary node $t(C)$ of a micro tree $C$ is connected to a boundary node in the parent micro tree $parent(C)$ (apart from the root).
The boundary node $t(C)$ might also be connected to a top boundary node of a child micro tree $child(C)$.\footnote{The root of the macro tree is an exception as it might have a top boundary node connected to two (rather than one) child micro trees. We focus on the other nodes. Handling the root is done in a very similar way.}
The bottom boundary node $b(C)$ of $C$ is connected to top boundary nodes of at most two child micro trees $\ell(C)$ and $r(C)$ of $C$.

\begin{figure}[h!]
   \centering
   \includegraphics[scale=0.6]{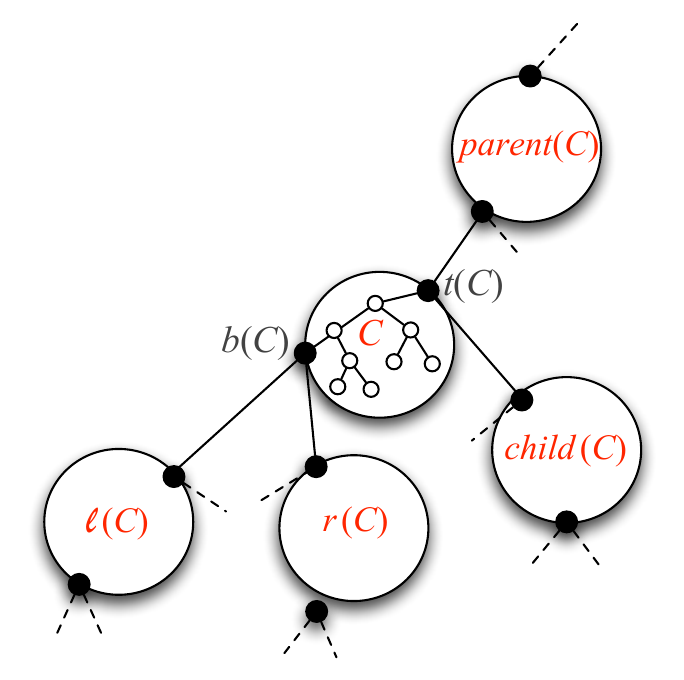}
   \caption{A micro tree $C$ and its neighboring micro trees in the macro tree. Inside each micro tree, the solid nodes correspond to boundary nodes and the hollow nodes to non-boundary nodes.}
   \label{fig:micromacro}
\end{figure}

\paragraph{\bf A bottom up traversal of the macro tree.}
With each micro tree $C$ we associate an array $A_C$. Let $T_C$ denote the union of micro tree $C$ and all its descendant micro trees (including the edges between them).
The array $A_C$ stores the maximum number of 1s (black nodes) in every pattern that includes the boundary node $t(C)$ and other nodes of $T_C$.
We also associate three auxiliary arrays: $A_b, A_t$ and $A_{tb}$.
The array $A_b$ stores the maximum number of 1s in every pattern that includes the boundary node $b(C)$ and possibly other nodes of $C$, $T_{\ell(C)}$, and $T_{r(C)}$. The array $A_t$ stores the maximum number of 1s in every pattern that includes the boundary node $t(C)$ and possibly other nodes of $C$ and $T_{child(C)}$. Finally, the array $A_{tb}$ stores the maximum number of 1s in every pattern that includes {\em both} boundary nodes $t(C)$ and $b(C)$ and possibly other nodes of $C$, $T_{\ell(C)}$, and $T_{r(C)}$.

We initialize for every micro tree $C$ its $O(|C|) =O(\log n)$ sized arrays. Arrays $A_C$ and $A_t$ are initialized to hold the maximum number of 1s in every pattern that includes $t(C)$ and nodes of $C$. This can be done in $O(|C|^2)$ time for each $C$ by rooting $C$ at $t(C)$ and running the algorithm from the previous subsection. Similarly, we initialize the array $A_b$ to hold the maximum number of 1s in every pattern that includes $b(C)$ and nodes of $C$. The array $A_{tb}$ is initialized as follows: First we check how many nodes are 1s and how many are 0s on the unique path between $t(C)$ and $b(C)$. If there are $i$ 1s and $j$ 0s we set $A_{tb}[k] =0$ for every $k< i+j$ and we set  $A_{tb}[i+j] = i$. We compute $A_{tb}[k]$ for all $k>i+j$ in total $O(|C|^2)$  time by contracting the $b(C)$-to-$t(C)$ path into a single node and running the previous algorithm rooting $C$ in this contracted node.
The total running time of the initialization step is therefore $O(n\cdot |C|^2 /\log n) = O(n \log n)$ which is negligible. Notice that during this computation we have computed the maximum number of 1s in all patterns that are completely inside a micro tree. We initialize the array $A$ (that is, only the first $\log n$ entries of $A$) with these values. In particular, this takes care of all patterns that do not contain any boundary node. 
We are now done with the leaf nodes of the macro tree.

We next describe how to compute the arrays of an internal node $C$ of the macro tree given the arrays of $\ell(C),r(C)$ and $child(C)$. We first compute the maximum number of 1s in all patterns that include $b(C)$ and possibly other vertices of $T_{\ell(C)}$ and $T_{r(C)}$. This can be done using the aforementioned string speedups in $O(|T_{\ell(C)}| \cdot |T_{r(C)}| / \log^2 n)$ time when both $|T_{\ell(C)}|>\log n$ and $|T_{r(C)}|>\log n$ and in $O(n / \log n)$ time otherwise.  
Using this and the initialized array $A_b$ of $C$ (that is of size $|C|\le \log n$)  we can compute the final array $A_b$ of $C$.  This is done by using the aforementioned string algorithm (on a string $S$ of length $|T_{\ell(C)}| + |T_{r(C)}|+1+ |C|$) restricted to the case where windows must include the split position (the split position separates $S$ to a substring of length $|T_{\ell(C)}| + |T_{r(C)}|+1$ and a substring of length $|C|\le \log n$). Using only the first speedup, this takes  time $O(|S| / \log n) = O(n/\log n)$.
Similarly, using the initialized $A_{tb}$ of $C$, we can compute the final array $A_{tb}$ of $C$ in  $O(n/\log n)$ time.

Next, we compute the array $A_t$ using the initialized array $A_t$ of $C$ and the array $A_t$ of $child(C)$ in time $O(n/ \log n)$. Finally, we compute $A_C$ of $C$ using $A_{tb}$ of $C$ and  $A_t$ of $child(C)$ in $O((|T_{\ell(C)}| + |T_{r(C)}| + 1+ |C|) \cdot |T_{child(C)}| / \log^2 n)$ time if both $|T_{\ell(C)}| + |T_{r(C)}| + 1+|C| >\log n$ and $|T_{child(C)}| > \log n$ and  in $O(n/ \log n)$ otherwise. To finalize $A_C$ we must then take the entry-wise maximum between the computed $A_C$ and $A_t$. This is because a pattern in $T_C$ may or may not include $b(C)$.
Finally, once $A_C$ is computed, we update the global array $A$ accordingly (by taking the entry-wise maximum between $A$ and $A_C$). 

To bound the total time complexity over all clusters $C$, notice that some computations required $O(\alpha(v)\cdot \beta(v)/ \log^2 n)$ when $\alpha(v)>\log n$ and $\beta(v)>\log n$ are the subtree sizes of two children of some node $v\in T$. We have already seen that the sum of all these terms over all nodes of $T$ is $O(n^2/\log^2 n)$. The other type of computations each require $O(n/\log n)$ time but there are at most $O(n/\log n)$ such computations ($O(1)$ for each micro tree) for a total of $O(n^2/\log^2 n)$.
This completes the proof of Theorem~\ref{thm:motifs}.

%% file: listing.tex
% !TEX root = jumbled.tex
\subsection{Finding the Query Pattern}\label{subsec:listing}

In this subsection we extend the index so that on top of identifying in $O(1)$ time if a pattern $(i,j)$ appears in $T$, it can also locate in $O(\log n)$ time  a node $v\in T$ that is part of such a pattern appearance. We call this node an {\em anchor} of the appearance. This extension increases the space of the index from $O(n)$ bits to $O(n\log n)$ bits (i.e., $\Oh{n}$ words).

Recall that given a tree $T$ we build in $O(n^2/\log ^2 n)$ time an array $A$ of size $n=|T|$ where $A[i]$ stores the minimum and maximum values $i_{min}$ and $i_{max}$ such that  $(i,i_{min})$ and $(i,i_{max})$ appear in $T$. Now consider a {\em centroid decomposition} of $T$: A centroid node $c$ in $T$ is a node whose removal leaves no connected component with more than $n/2$ nodes. We first construct the array $A$ of $T$ in $O(n^2/\log ^2 n)$ time and store it in node $c$. We then recurse on each remaining connected component. This way, every node $v\in T$ will compute the array corresponding to the connected component whose centroid was $v$. Notice that this array is not the array $A_v$ since we do not insist the pattern uses $v$. Observe that since each array $A$ is implicitly stored in an $n$-sized bit array $B$, and since the recursion tree is balanced the total space complexity is $O(n\log n)$ bits. Furthermore, since every node in $T$ has degree at most three, removing the centroid leaves at most three connected components and so the time to construct all the arrays is bounded by $T(n) = T(n_1) + T(n_2)+T(n_3) + O(n^2/\log ^2 n)$ where $n_1+n_2+n_3 = n$ and every $n_i \le n/2$. This yields the time complexity   $T(n)= O(n^2/\log ^2 n)$.

Let $c$ denote the centroid of $T$ whose removal leaves at most three connected components $T_1,T_2$, and $T_3$ (recall we assume degree at most 3).
Upon query  $(i,j)$ we first check the array of $c$ if pattern $(i,j)$ appears in $T$ (i.e., if $i_{min} \le j  \le i_{max}$). If it does
 then we check the centroids of $T_1,T_2$ and $T_3$. If $(i,j)$ appears in any of them then we continue the search there. This way, after at most $O(\log n)$ steps we reach the first node $v$ whose connected component includes $(i,j)$ but none of its child components do. We return $v$ as the anchor node since such a pattern must include $v$. 
We note that the above can be extended so that for {\em every} occurrence of $(i,j)$ one node that is part of this occurrence is reported. 
Finally, we note that it was recently observed in~\cite{SPIRE13} that if we are willing to settle for an index of size $O(n^2)$ then we can locate the entire match (not just an anchor) in time proportional to the size of the match.

%% file: grammars.tex
% !TEX root = jumbled.tex
\section{Jumbled Pattern Matching on Grammars} \label{sec:construction}

In grammar-based compression, a binary string $S$ of length $n$ is compressed using a context-free grammar $G(S)$ in Chomsky normal form that generates $S$ and only $S$. Such a grammar has a unique \emph{parse tree} that generates $S$. Identical subtrees of this parse tree indicate substring repeats in $S$. The size of the grammar $g=|G(S)|$ is defined as the total number of variables and production rules in the grammar. Note that $g$ can be exponentially smaller than $n=|S|$. We show how to solve the jumbled pattern matching problem on $S$ by solving it on the parse tree of $G(S)$, taking advantage of subtree repeats. We obtain the following bounds:

\begin{theorem} \label{thm:SLP}
Given a binary string $S$ of length $n$ compressed by a context free grammar $G(S)$ of size $g$, we can construct in $\Oh{g^{2 / 3} n^{4 / 3}/(\log n)^{4/3}}$ time a data structure of size $O(n)$ bits that on query $(i,j)$ determines in $O(1)$ time if $S$ has a substring of length $i$ with exactly $j$ 1s.
\end{theorem}

\begin{proof}
We will show how to compute the array $A$ such that $A[i]$ holds the maximum number of 1s in a substring of $S$ of size $i$. The minimum is found similarly.
We use a recent result of Gawrychowski~\cite{Gaw12} who showed how for any $\ell$, we can modify $G(S)$ in $\Oh{n}$ time by adding $\Oh{g}$ new variables such that every new variable generates a string of length at most $\ell$, and $S$ can be written as the concatenation of substrings generated by these $O(g)$ new variables. Thus, we can write $S$ as the concatenation of blocks $S=B_1 \circ \cdots \circ B_b$ with $b=\Oh{n/\ell}$ and $|B_j| \leq \ell$, such that amongst these blocks there are only $d=\Oh{g}$ distinct blocks \({B^*_1}, \ldots, {B^*_d}\). We refer to these $d$ blocks as \emph{basic blocks}.  
For each basic block ${B^*_k}$, \(1 \leq k \leq d\), we first build an array ${A^*_k}$ where $A^*_k[i]$ stores the maximum number of 1s over all substrings of $B^*_k$ of length $i$. This is done in $O(\ell^2 / \log^2 n)$ time per block (by using the  algorithm of~\cite{MR12} for strings) for a total of $O(g \cdot \ell^2 / \log^2 n)$.

We next handle substrings that span over two adjacent blocks. Namely, for each possible pair of basic blocks ${B^*_k}$ and ${B^*_m}$, \(1 \leq k \leq m \leq d\), we build a table ${A^*_{k, m}}$ where $A^*_{k, m}[i]$ stores the maximum number of 1s over all substrings of \({B^*_k \circ B^*_m}\) of length $i$ that start in ${B^*_k}$ and end in ${B^*_m}$. This is done in $O(\ell^2 /\log^2 n)$ time for each pair  for a total of $O(g^2\ell^2/\log^2 n)$. Recall that, since we use the  algorithm of~\cite{MR12}, the table ${A^*_{k, m}}$ is implicitly represented by an array of $O(\ell / \log n)$ words: For each $i$ that is a multiple of $\log n$, the array keeps one word storing the maximum number of 1s over all substrings of length $i$, and another word storing the binary increment vector for substrings of length $i+1,\ldots,i+ \log n -1$.

\begin{figure}[h!]
   \centering
   \includegraphics[scale=0.45]{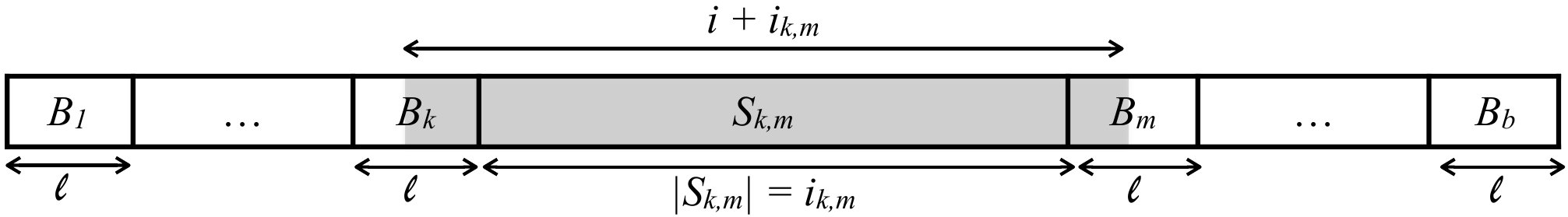}
   \caption{A string $S$ partitioned into blocks $B_1 \circ \cdots \circ B_b$, each of length at most $\ell$. The shaded substring is one of the  substrings considered for $A_{k, m}[i]$ as it includes the substring $S_{k, m}$ of length $i_{k, m}$ and a prefix and suffix of total length $i$.}
   \label{fig:grammar}
\end{figure}

Finally, we consider substrings that span over more than two blocks. 
For each pair of (non-basic) blocks ${B_k}$ and ${B_m}$, $1\le k < m \le b$, let 
$S_{k,m} = B_{k+1} \circ \cdots \circ B_{m-1}$ be a substring of $S$ that is of length $i_{k,m} = |S_{k,m}|$ and has $j_{k,m}$ $1$'s. Note that we can easily compute $i_{k,m}$ and $j_{k,m}$ of all $1\le k < m \le b$ in total time $O(n^2/\ell^2)$. 
For every $1\le k < m \le b$, we build a table $A_{k, m}$ of size $O(\ell)$ where $A_{k, m}[i]$ stores the maximum number of 1s over all substrings of \({B_k \circ \cdots \circ B_m}\) of length $i+i_{k,m}$ that start in ${B_k}$ and end in ${B_m}$. Notice that all such substring include $S_{k,m}$ as well as a suffix of $B_k$ and a prefix of $B_m$ whose total length is $i$ (see Fig.~\ref{fig:grammar}).
Therefore, for each $A_{k,m}$ we set $A_{k, m}[i]$ to be $j_{k,m}$ plus the maximal number of $1$'s in a suffix of $B_k$ and a prefix of $B_m$ whose total length is $i$. In other words, we set $A_{k, m}[i] = j_{k,m} + A^*_{k',m'}[i]$ where $k'$ (resp. $m'$) is such that the block ${B_k}$ (resp. ${B_m}$) corresponds to the basic block ${B^*_{k'}}$ (resp. ${B^*_{m'}}$). 
The computation of (an implicit representation of) each $A_{k, m}$ can be done in $O(\ell/ \log n)$ time by only setting $A_{k, m}[i]$ for $i$'s that are multiples of $\log n$ (the binary increment vectors of $A_{k, m}$ remain as in $A^*_{k', m'}$).
Since there are $O((n/\ell)^2)$ pairs of blocks and each pair requires $O(\ell/ \log n)$ time, we get a total of $O(n^2 / (\ell \log n))$ time.

Finally, once we have the implicit representation of all $A_{k, m}$'s we can compute the desired array $A$ from them in $O(n^2 / (\ell \log n))$ time: For each $i$ that is a multiple of $\log n$ and each $A_{k, m}$ we set $A[i]$ to be the maximum out of $A[i]$ and $A_{k, m}[i-i_{k,m}]$  in $O(1)$ time. The next $\log n$ entries of $A$ are computed in  $O(1)$ time (as done in~\cite{MR12}) from the increment vectors of $A[i]$ and $A_{k, m}[i-i_{k,m}]$.
To conclude, we get a total running time of $O(g^2\ell^2 /\log^2 n +n^2 / (\ell \log n))= \Oh{g^{2 / 3} n^{4 / 3}/(\log n)^{4/3}}$ when $\ell$ is chosen to be $(n/g)^{2/3} (\log n)^{1/3}$.
\qed
\end{proof}

\noindent We also note that similarly to the case of trees (Subsection~\ref{subsec:listing}), if we are willing to increase our index space to $O(n\log n)$ bits, then it is not difficult to turn indexes for {\em detecting} jumbled pattern matches in grammars into indexes for {\em locating} them. To obtain this, we build an index for $S$ and recurse (build indexes) on $S_1 = B_{1} \circ \cdots \circ B_{k}$ and $S_2 = B_{k+1} \circ \cdots \circ B_{d}$ where $|S_1|$ and $|S_2|$ are roughly $n/2$. This way, like in the centroid decomposition for trees, we can get in $O(\log n)$ time an anchor index of $S$. That is, an index of $S$ that is part of a pattern appearance. Furthermore, as opposed to trees, we can then find the actual appearance (not just the anchor) in additional $O(i)$ time by sliding a window of size $i$ that includes the anchor.

%% file: treewidth.tex
\section{Jumbled Pattern Matching on Bounded Treewidth Graphs}
\label{treewidth}

In this section we consider the extension of binary jumbled pattern matching to the domain of graphs: Given a graph $G$ whose vertices are colored either black or white, and a query $(i,j)$, determine whether $G$ has a connected subgraph $G'$ with $i$ white vertices and $j$ black vertices\footnote{The difference between the meaning of the query here and elsewhere in the paper is for ease of the presentation.}. This problem is also known as the (binary) graph motif problem in the literature. Fellows \emph{et al.}~\cite{FFHV11} provided an $n^{O(w)}$ algorithm for this problem, where $w$ is the treewidth of the input graph. Here we will substantially improve on this result by proving the following theorem, asserting that the problem is fixed-parameter tractable in the treewidth of the graph.
\begin{theorem}
\label{thm:treewidth}%
Binary jumbled pattern matching can be solved in $f(w)\cdot n^{O(1)}$ time on graphs of treewidth~$w$. The function $f(w)$ can be bounded by $w^{O(w)}$ in case a tree decomposition of width $w$ (see below) is provided with the input graph, and otherwise $f(w)=2^{O(w^3)}$.
\end{theorem}
Note that the algorithm in the theorem actually computes all queries $(i,j)$ that appear in $G$, and can thus be easily converted to an index for the input graph.\\

\noindent \textbf{Tree decompositions.} We begin by first introducing some necessary notation and terminology. Let $G=(V(G),E(G))$ be a graph. A \emph{tree decomposition} of $G$ is defined by a rooted tree $\mathcal{T}$ whose nodes are subsets of $V(G)$, called \emph{bags}, with the following two properties: $(i)$ the union of all subgraphs induced by the bags of $\mathcal{T}$ is $G$, and $(ii)$ for any vertex $x \in V(G)$, the set of all bags including $x$ induces a connected subgraph in $\mathcal{T}$. We use $\mathcal{X}$ to denote the set of bags in a given tree decomposition. The \emph{width} of the decomposition is defined as $\max_{X \in \mathcal{X}} |X|-1$. The \emph{treewidth} of $G$ is the smallest possible width of any tree decomposition of $G$. Given a bag $X$ of a given tree decomposition $\mathcal{T}$, we let $G_X$ denote the subgraph induced by the union of all bags in $\mathcal{T}_X$. Bodlaender~\cite{Bodlaender} gave an algorithm for computing a width-$w$ tree decomposition of a given graph with treewidth $w$ in $2^{O(w^3)} n$ time. We refer readers interested in further details to~\cite{DowneyFellows}.

We will work with a specific kind of tree decompositions, namely {\em nice tree decompositions}~\cite{Nice}. A nice tree decomposition is a binary rooted tree decomposition $\mathcal{T}$ with four types of bags: \emph{Leaf}, \emph{forget}, \emph{introduce}, and \emph{join}. Leaf bags are the leaves of $\mathcal{T}$ and are singleton sets which include a single vertex of $G$. A forget bag $X$ has one child $Y$ such that $X = Y \setminus \{x\}$ for some vertex $x$ of $G$. Thus,  $X$ \emph{forgets} the vertex $x$. Similarly, an introduce bag $X$ has one child $Y$ such that $X = Y \cup \{x\}$ for some vertex $x \notin Y$ of $G$. In this case, we say $X$ \emph{introduces} the vertex $x$. Finally a join bag $X$ has two children $Y$ and $Z$ in $\mathcal{T}$ with $X=Y=Z$. It is well known that given a tree decomposition of any graph, one can compute in polynomial-time a nice tree decomposition of the same graph with equal width and with an at most linear increase in its number of nodes~\cite{Nice}. Thus, from this point onwards we may assume that we are given a nice tree decomposition $\mathcal{T}$ of $G$ with width $w$. \\

\noindent \textbf{Positive partitions.} We next describe the main data structure that we compute in our algorithm. Let $X$ be an arbitrary bag. A partition\footnote{Here we slightly abuse our terminology and allow $X_0$ to be the empty set. } $\Pi_X=\{X_0,X_1,\ldots,X_x\}$ of $X$ is \emph{positive} for a given query $(i,j)$ in $G_X$ if there are $x$ disjoint connected subgraphs $G_1,\ldots,G_x$  of $G_X$ such that $(1)$ the total number of black (reps. white) vertices in $G'=G_1 \cup\cdots\cup G_x$  is $i$ (resp.  $j$), and $(2)$  $V(G') \cap X_0 = \emptyset$ and $V(G_\ell) \cap X = X_\ell$ for each $\ell=1,\ldots,x$ (see Fig.~\ref{fig:partition}). Thus, positive partitions capture \emph{partial occurrences} that intersect $X$ at exactly $X \setminus X_0$. These may not be actual occurrences as we do not require any edges between the different $G_i$'s, and so $G'$ itself may not be connected. We let $A_X[i,j]$ denote the set of all positive partitions for a query $(i,j)$, and let $A_X$ denote the array with an entry for each possible query $(i,j)$. We will require that the trivial partition where $X_0=X$ is only positive for the query $(0,0)$.

\begin{figure}[h!]
\centering
\includegraphics[scale=0.42]{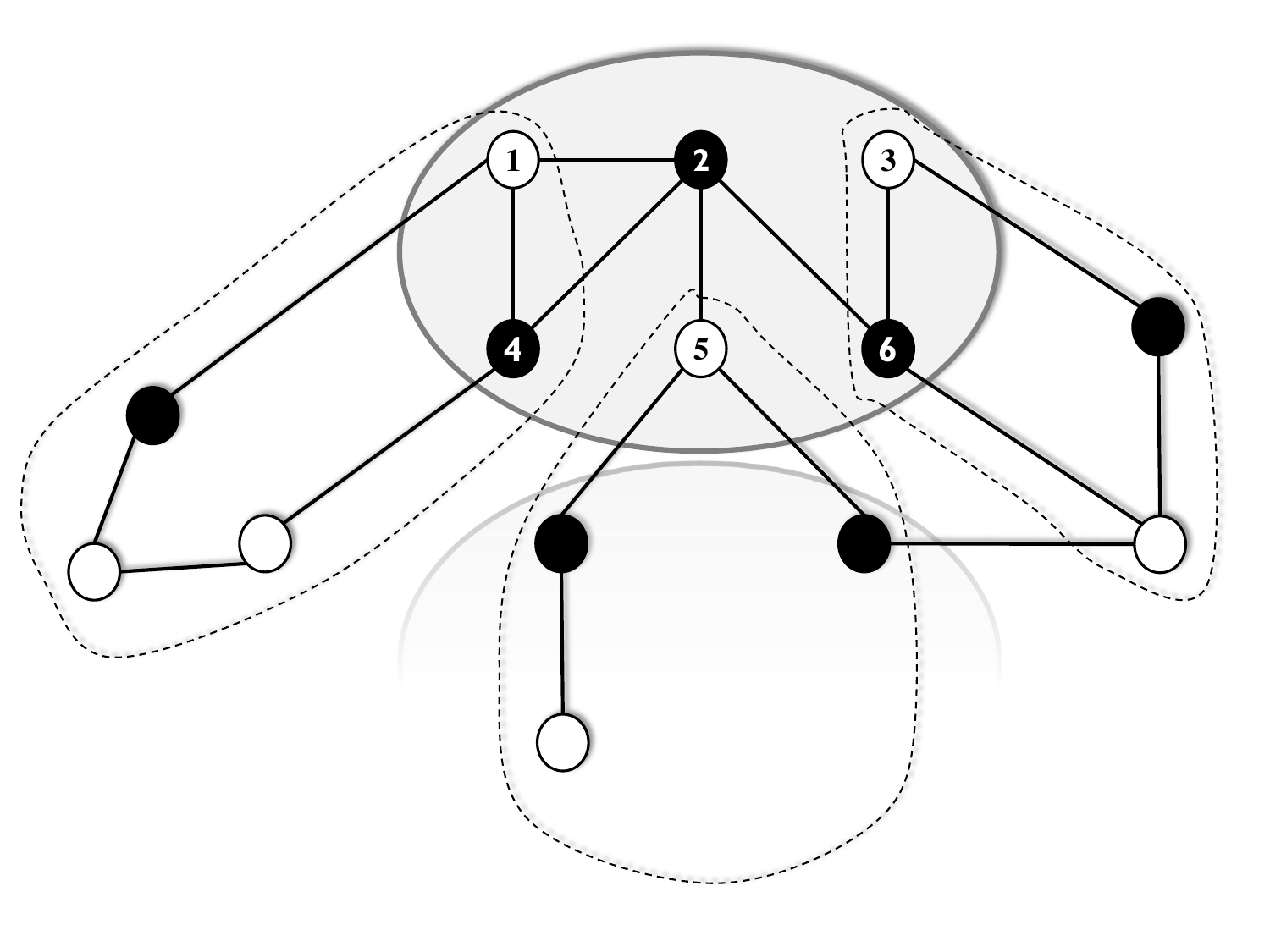}
\caption{A positive partition for the query $(7,6)$ in $G_X$, where $X=\{1,2,3,4,5,6\}$. The partition is defined by $\big\{\{2\}, \{1,4\}, \{5\},\{3,6\}\big\}$. Notice that vertex 2 is not in any of the connected graphs witnessing the partial occurrence. Also note that these graphs may or may not have edges between them.}
\label{fig:partition}
\end{figure}

Note that by definition, a query $(i,j)$ appears in $G_X$ iff there exists some partition into two sets $\{X_0,X_1\}$ that is positive for $(i,j)$ in $G$. Since $(i,j)$ appears in $G$ iff $(i,j)$ appears in $G_X$ for some bag $X \in \mathcal{X}$, this means that it is also positive for $(i,j)$ in some $G_X$. Thus, computing the arrays $A_X$ for all bags $X \in \mathcal{X}$ suffices for solving our problem. We do this by computing all arrays $A_X$ in a bottom-top fashion from the leaves to the root of $\mathcal{T}$. Note that the size of each array $A_X$ can easily be bounded by $w^{O(w)} n^2$, considering that the $w$'th Bell number is bounded by $w^{O(w)}$. Thus, to get a similar term in our running time, we will show that computing the array $A_X$ from the arrays of the children of $X$ can be done in polynomial-time. The computation on leaf bags is trivial (as they are singletons), and the computation on forget nodes is almost equally easy: If $X$ is a forget bag with child $Y$, then computing $A_X$ from $A_Y$ in this case amounts to converting each positive partition $\Pi_Y$ of $Y$ to a corresponding positive partition $\Pi_X$ of $X$ by removing $x$, the vertex forgotten by $X$, from the class it belongs to in $\Pi_Y$. We thus focus below on introduce nodes and join nodes.\\

\noindent \textbf{Introduce nodes:}  Let $X$ be an introduce bag with child $Y$ in $\mathcal{T}$, and let $x$ be the vertex introduced by $X$. Let us assume for ease of presentation that $x$ is colored white (the case where it is colored black is symmetric). By the properties of a tree decomposition, we know that $x$ is only adjacent to vertices $y \in Y$ in $G_X$~\cite{DowneyFellows}. Let $y_1,\ldots,y_\ell$ denote these neighbors of $x$, and let $G^k_X$ denote the graph obtained by deleting the edges $\{x,y_{k+1}\},\ldots,\{x,y_\ell\}$ from $G_X$ for each $k=0,\ldots,\ell$ ($G^\ell_X =G_X$). Similarly, let $A^k_X[i,j]$ denote the set of all positive partitions of $(i,j)$ in $G^k_X$. We will compute $A^0_X[i,j]$ from $A_Y$, and $A^k_X[i,j]$ from $A^{k-1}_X[i,j]$ for each $k > 0$. Finally, we will set $A_X[i,j] = A^\ell_X[i,j]$.

We begin with $k=0$. In this case, $x$ is an isolated vertex in $G^0_X$. Hence, there are only two types of positive partitions of $(i,j)$ in $G_X$:
\begin{itemize}
\item A partition $\Pi_X$ obtained by taking $\Pi_X= \Pi_Y \cup \{\{x\}\}$ for some $\Pi_Y \in A_Y[i-1,j]$ (thus, $\{x\}$ is a singleton set in $\Pi_X$).
\item A partition $\Pi_X$ obtained by taking a partition $\Pi_Y \in A_Y[i,j]$ and adding $x$ to $Y_0 \in \Pi_Y$ (thus, $x$ is in the set of vertices not included in the partial occurrence captured by $\Pi_X$).
\end{itemize}
It is easy to see that since $x$ is an isolated vertex the above description indeed captures all types of positive partitions for $G^0_X$, and so we can compute $A^0_X[i,j]$ from $A_Y$ in polynomial time.

Assume that  $k > 0$. Then any positive partition for $(i,j)$ in $G^{k-1}_X$ is also positive in $G^k_X$. Moreover, the only new positive partitions for $(i,j)$ in $G^k_X$ that were not positive in $G^{k-1}_X$ are partitions where $x$ and $y_k$ belong to the same class (although, there might be partitions of this type which were positive in $G^{k-1}_X)$. Thus, we compute $A^k_X[i,j]$ by first setting $A^k_X[i,j]=A^{k-1}_X[i,j]$. Then for each $\Pi \in A_X^{k-1}[i,j]$ with $x \in X_i \in \Pi$ and $y_k \in X_j \in \Pi$, $i \neq j$, we add the partition $\big(\Pi \setminus \{X_i,X_j\}\big) \, \cup \, \{X_i \cup X_j\}$ to $A^k_X[i,j]$ (assuming it is not already there). The total amount of computation time required here is obviously polynomial in the size of $A^{k-1}_X[i,j]$.\\

\noindent \textbf{Join nodes:} Consider a join bag $X$ with two children $Y$ and $Z$ in $\mathcal{T}$, and recall that $X=Y=Z$. For a pair of partitions $\Pi_Y=\{Y_0,\ldots,Y_y\}$ and $\Pi_Z=\{Z_0,\ldots,Z_z\}$ of $Y$ and $Z$, we define the partition $\Pi_Y \oplus \Pi_Z$ (the \emph{join} of $\Pi_Y$ and $\Pi_Z$) as follows: First we set $X_0$ to be $Y_0 \cap Z_0$. The remaining classes are constructed such that any pair of vertices in $X$ belong to the same class in $\Pi_X \setminus \{X_0\}$ iff they belong to the same class in $\Pi_Y \setminus \{Y_0\}$ or to the same class in $\Pi_Z \setminus \{Z_0\}$. Thus, the equivalence relation defined by $\Pi_X \setminus X_0$ is the transitive closure of the union of the two equivalence relations defined by $\Pi_Y \setminus \{Y_0\}$ and $\Pi_Z \setminus \{Z_0\}$.

Let $i_0$ and $j_0$ respectively denote the number of white and black vertices in $X$. We claim that if $(i_1,j_1)$ and $(i_2,j_2)$ are two  queries for which $\Pi_Y$ and $\Pi_Z$ are respectively positive in $G_Y$ and $G_Z$, then $\Pi_X=\Pi_Y \oplus \Pi_Z$ is positive for $(i_1+i_2-i_0,j_1+j_2-j_0)$. This can be verified by considering the connected components in the graph $G'_X = G^Y_1 \cup \cdots \cup G^Y_y \cup G^Z_1 \cdots \cup G^Z_z$, where $G^Y_1,\ldots,G^Y_y$ and $G^Z_1,\ldots,G^Z_z$ are sets of graphs witnessing that $\Pi_Y$ and $\Pi_Z$ are positive for $(i_1,j_1)$ in $G_Y$ and $(i_2,j_2)$ in $G_Z$. It is easy to see that the total number of white and black vertices in these components is $i=i_1+i_2-i_0$ and $j=j_1+j_2-j_0$, where $i_0$ white vertices and $j_0$ black vertices are subtracted due to double counting the vertex colors in $X$. Moreover, it can be verified that these components intersect $X$ as required by~$\Pi_X$, due to the fact that $\Pi_X$ is the transitive closure of $\Pi_Y \cup \Pi_Z$. Thus, $G'_X$ is a partial occurrence of $(i,j)$ in $G_X$, and $\Pi_X \in A_X[i,j]$.

On the other hand, it can also be seen on the same lines that if $(i,j)$ is a query for which $\Pi_X$ is positive in $G_X$, then it is either in $\Pi_Y[i,j]$ or $\Pi_Z[i,j]$, or we have $(i,j)=(i_1+i_2-i_0,j_1+j_2-j_0)$ for some pair of queries $(i_1,j_1)$ and $(i_2,j_2)$ for which $\Pi_Y$ and $\Pi_Z$ are positive in $G_Y$ and $G_Z$. We can therefore compute $A_X[i,j]$ by first setting $A_X[i,j]$, and then examining all pairs $(i_1,j_1)$ and $(i_2,j_2)$ as above. For each such pair, we compute all partitions $\Pi_Y \oplus \Pi_Z$ for $\Pi_Y \in A_Y[i_1,j_1]$ and $\Pi_Z \in A_Z[i_2,j_2]$. Note that the entire computation of $A_X$ requires time which is polynomial in the total sizes of $A_Y$ and $A_Z$.\\

\noindent \textbf{Summary.} We have shown above how to compute, for any given query $(i,j)$, the array $A_X$ for each bag $X$ of $\mathcal{T}$ in $w^{O(w)}n^{O(1)}$ time. As the total number of bags is $O(n)$, we obtain an algorithm whose total running time is $w^{O(w)}n^{O(1)}$, excluding the time required to compute the nice tree decomposition $\mathcal{T}$. This completes the proof of Theorem~\ref{thm:treewidth}. We note that our algorithm straightforwardly extends to an $w^{O(w)}n^{O(c)}$ time algorithm for the case where the vertices of $G$ are colored with $c$ colors.

%% file: jumbled.bbl
\begin{thebibliography}{10}

\bibitem{MicroMacro}
S.~Alstrup, J.~Secher, and M.~Sporkn.
\newblock Optimal on-line decremental connectivity in trees.
\newblock {\em Information Processing Letters}, 64(4):161--164, 1997.

\bibitem{AmbalathBHKMPR10}
A.M. Ambalath, R.~Balasundaram, C.H. Rao, V.~Koppula, N.~Misra, G.~Philip, and
  M.S. Ramanujan.
\newblock On the kernelization complexity of colorful motifs.
\newblock In {\em Proceedings of the 5th International Symposium Parameterized
  and Exact Computation}, pages 14--25, 2010.

\bibitem{Amir}
A.~Amir, T.M. Chan, M.~Lewenstein, and N.~Lewenstein.
\newblock On hardness of jumbled indexing.
\newblock In {\em Proceedings of the 41st International Colloquium on Automata,
  Languages and Programming (ICALP)}, 2014.
\newblock To appear.

\bibitem{BFKL12}
G.~Badkobeh, G.~Fici, S.~Kroon, and Z.~Lipt\'ak.
\newblock Binary jumbled string matching for highly run-length compressible
  texts.
\newblock {\em Information Processing Letters}, 113(17):604--608, 2013.

\bibitem{Alignment}
G.~Benson.
\newblock Composition alignment.
\newblock In {\em Proceedings of the 3rd International Workshop on Algorithms
  in Bioinformatics (WABI)}, pages 447--461, 2003.

\bibitem{BetzlerBFKN11}
N.~Betzler, R.~van Bevern, M.R. Fellows, C.~Komusiewicz, and R.~Niedermeier.
\newblock Parameterized algorithmics for finding connected motifs in biological
  networks.
\newblock {\em IEEE/ACM Trans. Comput. Biology Bioinform.}, 8(5):1296--1308,
  2011.

\bibitem{SNP}
S.~B{\"o}cker.
\newblock Simulating multiplexed {SNP} discovery rates using base-specific
  cleavage and mass spectrometry.
\newblock {\em Bioinformatics}, 23(2):5--12, 2007.

\bibitem{Bodlaender}
H.L. Bodlaender.
\newblock A linear time algorithm for finding tree-decompositions of small
  treewidth.
\newblock {\em SIAM Journal on Computing}, 25:1305--1317, 1996.

\bibitem{Nice}
H.L. Bodlaender.
\newblock Treewidth. {A}lgorithmic techniques and results.
\newblock In {\em Proceedings of the 22nd international symposium on
  Mathematical Foundations of Computer Science (MFCS)}, pages 19--36, 1997.

\bibitem{BCFL10}
P.~Burcsi, F.~Cicalese, G.~Fici, and Z.~Lipt\'ak.
\newblock On table arrangement, scrabble freaks, and jumbled pattern matching.
\newblock In {\em Proceedings of the Symposium on Fun with Algorithms}, pages
  89--101, 2010.

\bibitem{BCFL12a}
P.~Burcsi, F.~Cicalese, G.~Fici, and Z.~Lipt\'ak.
\newblock Algorithms for jumbled pattern matching in strings.
\newblock {\em International Journal of Foundations of Computer Science},
  23(2):357--374, 2012.

\bibitem{BCFL12b}
P.~Burcsi, F.~Cicalese, G.~Fici, and Z.~Lipt\'ak.
\newblock On approximate jumbled pattern matching in strings.
\newblock {\em Theory of Computing Systems}, 50(1):35--51, 2012.

\bibitem{Charikaretal2005}
M.~Charikar, E.~Lehman, D.~Liu, R.~Panigrahy, M.~Prabhakaran, A.~Sahai, and
  A.~Shelat.
\newblock The smallest grammar problem.
\newblock {\em IEEE Transactions on Information Theory}, 51(7):2554--2576,
  2005.

\bibitem{CFL09}
F.~Cicalese, G.~Fici, and Z.~Lipt\'ak.
\newblock Searching for jumbled patterns in strings.
\newblock In {\em Proceedings of the Prague Stringology Conference}, pages
  105--117, 2009.

\bibitem{SPIRE13}
F.~Cicalese, T.~Gagie, E.~Giaquinta, E.~Laber, S.~Liptak, R.~Rizzi, and A.~I.
  Tomescu.
\newblock Indexes for jumbled pattern matching in strings, trees and graphs.
\newblock In {\em Proceedings of the 20th International Symposium on String
  Processing and Information Retrieval (SPIRE)}, pages 56--63, 2013.

\bibitem{CLWY12}
F.~Cicalese, E.~S. Laber, O.~Weimann, and R.~Yuster.
\newblock Near linear time construction of an approximate index for all maximum
  consecutive sub-sums of a sequence.
\newblock In {\em Proceedings of the Symposium on Combinatorial Pattern
  Matching}, pages 149--158, 2012.

\bibitem{DondiFV11}
R.~Dondi, G.~Fertin, and S.~Vialette.
\newblock Complexity issues in vertex-colored graph pattern matching.
\newblock {\em J. Discrete Algorithms}, 9(1):82--99, 2011.

\bibitem{DondiFV11b}
R.~Dondi, G.~Fertin, and S.~Vialette.
\newblock Finding approximate and constrained motifs in graphs.
\newblock In {\em Proceedings of the 22nd Annual Symposium Combinatorial
  Pattern Matching}, pages 388--401, 2011.

\bibitem{DowneyFellows}
R.G. Downey and M.R. Fellows.
\newblock {\em Parameterized complexity}.
\newblock Springer, 1999.

\bibitem{FFHV11}
M.R. Fellows, G.~Fertin, D.~Hermelin, and S.~Vialette.
\newblock Upper and lower bounds for finding connected motifs in vertex-colored
  graphs.
\newblock {\em J. Comput. Syst. Sci.}, 77(4):799--811, 2011.

\bibitem{Gaw12}
P.~Gawrychowski.
\newblock Faster algorithm for computing the edit distance between
  {SLP}-compressed strings.
\newblock In {\em Proceedings of the Symposium on String Processing and
  Information Retrieval}, pages 229--236, 2012.

\bibitem{GG12}
E.~Giaquinta and S.~Grabowski.
\newblock New algorithms for binary jumbled pattern matching.
\newblock {\em Information Processing Letters}, 113(14-16):538--542, 2013.

\bibitem{Jacobson}
G.~Jacobson.
\newblock Space-efficient static trees and graphs.
\newblock In {\em Proceedings of the 30th Annual Symposium on Foundations of
  Computer Science (FOCS)}, pages 549--554, 1989.

\bibitem{ESA13}
T.~Kociumaka, J.~Radoszewski, and W.~Rytter.
\newblock Efficient indexes for jumbled pattern matching with constant-sized
  alphabet.
\newblock In {\em ESA}, pages 625--636, 2013.

\bibitem{LFS06}
V.~Lacroix, C.G. Fernandes, and M.-F. Sagot.
\newblock Motif search in graphs: Application to metabolic networks.
\newblock {\em IEEE/ACM Trans. Comput. Biology Bioinform.}, 3(4):360--368,
  2006.

\bibitem{MR10}
T.~M. Moosa and M.~S. Rahman.
\newblock Indexing permutations for binary strings.
\newblock {\em Information Processing Letters}, 110(18--19):795--798, 2010.

\bibitem{MR12}
T.~M. Moosa and M.~S. Rahman.
\newblock Sub-quadratic time and linear space data structures for permutation
  matching in binary strings.
\newblock {\em Journal of Discrete Algorithms}, 10:5--9, 2012.

\bibitem{Munro}
I.~Munro.
\newblock Tables.
\newblock In {\em Proceedings of the 16th Foundations of Software Technology
  and Theoretical Computer Science (FSTTCS)}, pages 37--42, 1996.

\bibitem{Ryt03}
W.~Rytter.
\newblock Application of {Lempel-Ziv} factorization to the approximation of
  grammar-based compression.
\newblock {\em Theoretical Computer Science}, 302(1--3):211--222, 2003.

\end{thebibliography}
